\documentclass[nohyperref]{article}

\usepackage{microtype}
\usepackage{graphicx}
\usepackage{subfigure}
\usepackage{booktabs} 

\usepackage{hyperref}

\usepackage[ruled,vlined]{algorithm2e}
\usepackage[accepted]{icml2022}

\usepackage{amsmath}
\usepackage{amssymb}
\usepackage{mathtools}
\usepackage{amsthm}
\usepackage{multirow}
\usepackage[capitalize,noabbrev]{cleveref}

\theoremstyle{plain}
\newtheorem{theorem}{Theorem}[section]
\newtheorem{proposition}[theorem]{Proposition}
\newtheorem{lemma}[theorem]{Lemma}

\theoremstyle{definition}
\newtheorem{definition}[theorem]{Definition}

\theoremstyle{remark}
\newtheorem{remark}[theorem]{Remark}

\usepackage[textsize=tiny]{todonotes}

\icmltitlerunning{RetrievalGuard: Provably Robust 1-Nearest Neighbor Image Retrieval}

\begin{document}

\twocolumn[
\icmltitle{RetrievalGuard: Provably Robust 1-Nearest Neighbor Image Retrieval}




\begin{icmlauthorlist}
\icmlauthor{Yihan Wu}{pitt}
\icmlauthor{Hongyang Zhang}{waterloo}
\icmlauthor{Heng Huang}{pitt}
\end{icmlauthorlist}

\icmlaffiliation{pitt}{Department of Electrical and Computer Engineering, University of Pittsburgh, USA}
\icmlaffiliation{waterloo}{David R. Cheriton School of Computer Science, University of Waterloo, Canada}

\icmlcorrespondingauthor{Yihan Wu}{yiw154@pitt.edu}
\icmlcorrespondingauthor{Hongyang Zhang}{hongyang.zhang@uwaterloo.ca}
\icmlcorrespondingauthor{Heng Huang}{henghuanghh@gmail.com}

\icmlkeywords{Adversarial Robustness}

\vskip 0.3in
]



\printAffiliationsAndNotice{} 

\begin{abstract}
 Recent research works have shown that image retrieval models are vulnerable to adversarial attacks, where slightly modified test inputs could lead to problematic retrieval results.  In this paper, we aim to design a provably robust image retrieval model which keeps the most important evaluation metric Recall@1 invariant to adversarial perturbation. We propose the first 1-nearest neighbor (NN) image retrieval algorithm, RetrievalGuard, which is provably robust against adversarial perturbations within an $\ell_2$ ball of calculable radius. 
 The challenge is to design a provably robust algorithm that takes into consideration the 1-NN search and the high-dimensional nature of the embedding space.
 Algorithmically, given a base retrieval model and a query sample, we build a smoothed retrieval model by carefully analyzing the 1-NN search procedure in the high-dimensional embedding space. We show that the smoothed retrieval model has bounded Lipschitz constant and thus the retrieval score is invariant to $\ell_2$ adversarial perturbations. Experiments on image retrieval tasks validate the robustness of our RetrievalGuard method.
\end{abstract}

\section{Introduction}

Image retrieval has been an important and active research area in computer vision with broad applications, such as person re-identification \cite{zheng2015scalable}, remote sensing \cite{chaudhuri2019siamese}, medical image search \cite{nair2020review}, and shopping recommendation \cite{liu2016deepfashion}.
In a typical image retrieval task, given a query image, the image retrieval algorithm selects semantically similar images from a large gallery. To conduct efficient retrieval, the high-dimensional images are often encoded into an embedding space by deep neural networks (DNNs). The encoder is expected to cluster semantically similar images while separating dissimilar images.

Despite a large amount of works on image retrieval, many fundamental questions remain unresolved. For example, DNNs are notorious for their vulnerability to adversarial examples \cite{szegedy2013intriguing,biggio2013evasion,yang2020design,blum2022boosting,zhang2022many}, \emph{i.e.} slightly modified test inputs can lead to the largely changed and incorrect prediction results.
In image retrieval, the encoders are oftentimes parameterized by DNNs and existing approaches are susceptible to adversarial attacks. Though adversarial training alleviates the issue by building a backbone that is robust against off-the-shelf attacks~\cite{zhang2019theoretically}, the backbone is not certifiably robust against attacks with growing power. In fact, there has been long-standing arms race between adversarial defenders and attackers: defenders design empirically robust algorithms which are later exploited by new attacks designed to undermine those defenses~\cite{athalye2018obfuscated}. Moreover, existing defenses \cite{panum2021exploring,zhou2020adversarial} only focus on one type of attacks and fail to generalize to other types of attacks.
Thus, it is desirable to develop more powerful defenses for image retrieval with \textbf{provable} adversarial robustness.

For image classification tasks, there are two types of provably robust methods against adversarial perturbations. The first type is deterministic approaches represented by linear relaxations \cite{zhang2020towards}, mixed-integer linear programming \cite{tjeng2018evaluating}, and Lipschitz constant estimation \cite{zhang2019recurjac}. But these approaches only work with certain neural architectures and are hard to train. The second category is randomized smoothing  \cite{cohen2019certified,li2019certified}, which provides probabilistic robustness guarantees. The insight behind randomized smoothing is a construction of smoothed model $g$ by voting the prediction of vanilla model $h$ over a smoothing distribution. The smoothed model $g$ is provably Lipschitz bounded. Additionally, randomized smoothing is model-agnostic and can be applied to arbitrary backbones. Although recent works show that randomized smoothing suffers from curse of dimensionality \cite{blum2020random,kumar2020curse,wu2021completing}, the method remains the state-of-the-art certified defense against adversarial perturbation.

\noindent\textbf{Challenges.}
Unfortunately, direct application of randomized smoothing does not work in our setting of image retrieval. Randomized smoothing is carefully designed for classification tasks, where the output of the model is a discrete label. However, in the image retrieval task, the output of the model is a high-dimensional embedding vector. It remains unclear how to implement the ``voting'' operation for the embedding vector. In addition, retrieval results are computed by comparing the distance between the embedding of query images and gallery images and finding the nearest neighbor, but randomized smoothing was not designed for this procedure. Therefore, many observations and techniques for randomized smoothing break down when we consider more sophisticated image retrieval tasks.

\noindent\textbf{Our setting.} 
In the image retrieval tasks, we search for semantically similar images in a large reference set for a given query image. The quality of an image retrieval model can be measured by the Recall@k score: given a reference set $R$, a query sample $x$ and an embedding mapping $h(\cdot)$, the Recall@k of sample $x$ is 1 if the first $k$ nearest neighbors of $h(x)$ in $R$ contains at least one sample with the same class as $x$; otherwise, Recall@k of $x$ is 0. We expect the retrieval score to be 1 for as many query samples as possible. Among Recall@k, perhaps the most widely-used metric is Recall@1. Our goal is to design a provably robust retrieval model which keeps the metric Recall@1 invariant to adversarial perturbation, \emph{i.e.}, the nearest neighbor of $x$ is of the same class as $x$ even in the presence of $\ell_2$ bounded perturbations.

\medskip
\noindent\textbf{Summary of contributions.}
Our work explores the adversarial robustness of 1-NN image retrieval.
\begin{itemize}
    \item
    Algorithmically, we propose RetrievalGuard, the first provably robust 1-NN image retrieval framework against $\ell_2$ adversarial perturbations. Given an input and a base embedding mapping, our algorithm averages the embeddings of Gaussian-perturbed inputs to achieve the robustness and conducts 1-NN search based on the smoothed embedding.
    \item
    Theoretically, we analyze     RetrievalGuard by new proof techniques regarding the 1-NN search and the smoothed high-dimensional embedding. We show that the smoothed embedding is Lipschitz with a tight and calculable Lipschitz bound. Additionally, we analyze the Monte-Carlo method for computing the certified radius of each input. The algorithmic error only logarithmically depends on the dimension of the embedding space. Our analysis of smoothed embedding might be of independent interest to other computer vision tasks more broadly.
    \item
    Experimentally, we evaluate the certified robustness and accuracy of RetrievalGuard on popular image retrieval benchmarks under different choices of dimension of embedding space, number of Monte-Carlo samplings, and variance of Gaussian noise.
\end{itemize}

\section{Related Works}

  \noindent\textbf{Deep metric learning.} 
  Deep metric learning (DML) is one of the most popular methods used for image retrieval. It learns semantic embedding of images by putting the feature vectors of similar samples closer in the embedding space while separating the feature vectors of dissimilar samples.
  There are two types of metric losses in DML, tuple-based loss and classification-based loss. Tuple-based loss characterizes the distance between similar and dissimilar image embedding, which includes triplet loss \cite{schroff2015facenet}, margin loss \cite{wu2017sampling}, and multi-similarity loss \cite{wang2019multi}. 
Classification-based loss is designed with a fixed \cite{boudiaf2020unifying} or learnable proxy \cite{kim2020proxy}, where the proxy refers to a subset of training data. 
However, the performance of different DML losses are similar under the same training settings~\cite{roth2020revisiting,musgrave2020metric}. In this work, we choose a broadly used DML model, DML with margin loss~\cite{wu2017sampling}, as the base image retrieval model in our experiments.

\medskip
\noindent\textbf{Image retrieval attacks.} \cite{bouniot2020vulnerability,wang2020transferable} designed metric-based attacks for person re-identification tasks, where the adversarial samples were generated by maximizing the distance between similar pairs and minimizing the distance between dissimilar pairs. 
\cite{feng2020adversarial} attacked a type of image retrieval method, deep product quantization network, by
generating perturbations from the peak of the Centroid Distribution, which is the estimation of the probability distribution of codewords assignment.
\cite{li2019universal} introduced a universal perturbation attack on image retrieval to break the neighborhood relationships of image features via degrading the corresponding ranking metric.
\cite{zhou2020adversarial} proposed image ranking candidate attack and query attack, which can raise or lower the rank of selected candidates by adversarial perturbations.
  
\medskip
\noindent\textbf{Randomized smoothing.} If the prediction of a model on sample $x$ does not change in the presence of perturbations with bounded radius $r$, this model is said to be certifiably robust on sample $x$ with radius $r$. 
To the best of our knowledge, randomized smoothing \cite{cohen2019certified,Lecuyer2019,li2019certified,salman2019provably} is currently the only approach that provides certified robustness in a model-agnostic way. Applications of randomized smoothing include image classification \cite{cohen2019certified}, graph classification \cite{bojchevski2020efficient}, and point cloud classification \cite{liu2021pointguard}, with $\ell_0$, $\ell_2$ and $\ell_\infty$ robustness guarantees.
In the image classification tasks, randomized smoothing transforms a base classifier $f$ to a smoothed classifier $g$, which is certifiably robust in an $\ell_2$ ball. More specifically, given a sample $x$ and arbitrary binary classifier $f$ which maps inputs in $\mathbb{R}^d$ to a class in $\{0,1\}$, the smoothed classifier $g$ labels $x$ as the majority vote of predictions of $f$ on the Gaussian-perturbed images $\mathcal{N}(x,\sigma^2I_d)$. In particular, let
$$g(x) = \mathbb{P}(\{z|f(x+z)=c_x\}),$$ where $c_x$ is the  label of sample $x$,
  we expect $g(x)>0.5$ to correctly classify $x$. An important property of the smoothed classifier $g$ is its $L$-Lipschitzness \emph{w.r.t.} $\ell_2$ norm \cite{salman2019provably}. With this property, for arbitrary perturbation $\delta$ such that $\|\delta\|_2<r$, the difference between two prediction scores $g(x)$ and $g(x+\delta)$ is bounded by $L\|x-(x+\delta)\|_2<Lr$. Thus 
  if $g(x)>0.5$, we can choose a small $r$, namely, $r=(g(x)-0.5)/2L$, such that $g(x+\delta)>0.5$ for all perturbations $\delta$ with $\|\delta\|_2<r$. In \cite{cohen2019certified}, the authors obtained a tighter certified radius $r = \sigma\Phi^{-1}(\underline{g(x)})$, where $\underline{g(x)}$ is a probabilistic lower bound of $g(x)$ and $\Phi$ is the cumulative distribution function of standard Gaussian distribution.

\medskip
\noindent\textbf{Novelties and difference of our method from randomized smoothing.}
In this work, we focus on a different problem from classification, namely, 1-NN image retrieval. Unlike the binary classification tasks, where the output of the base model $f: \mathbb{R}^d\to \{0,1\}$ is a discrete one-dimensional scalar, the base model in the 1-NN retrieval task is an embedding model $h: \mathbb{R}^d\to \mathbb{R}^k$ with a high-dimensional output. Therefore, directly applying randomized smoothing to our image retrieval problem does not work. Instead, we propose a new proof technique and demonstrate that we can build a smoothed embedding $g:\mathbb{R}^d\to \mathbb{R}^k$ that is Lipschitz continuous. Algorithimcally, different from voting by majority as in randomized smoothing, our algorithm is built upon averaging the embedding of Gaussian-perturbed inputs. We carefully analyze the nearest neighbor of a query sample in the positive and negative reference sets of embedding space, such that the nearest neighbor is stable to adversarial perturbation in the input space. Our analysis of smoothed embedding might be of independent interest to other representation learning tasks more broadly.

  \section{RetrievalGuard}
  In this section, we will introduce our method of building a certifiably robust embedding for the 1-NN image retrieval task. Denote the embedding model by $h$. For each sample $x$, we will first calculate its embedding $h(x)$. We then search for the sample $x'$ in the reference set whose embedding $h(x')$ is closest to $h(x)$. If the ground-truth labels of $x$ and $x'$ match, the retrieval score is 1; otherwise, the retrieval score is 0. Our goal is to build a certifiably robust retrieval model, such that the retrieval score is invariant to arbitrary $\ell_2$ bounded perturbations. All proofs of this section can be found in the Appendix.
  
    \medskip
  \noindent\textbf{Intuition of RetrievalGuard.} 
  RetrievalGuard is an approach to build a provably robust image retrieval from a vanilla image retrieval model.
  In RetrievalGuard, we will build a smoothed retrieval model by averaging the embedding of the given model, and calculate the robustness guarantee for the smoothed model based on its Lipschitz continuous property. We want to emphasize that the given model doesn't have any robustness guarantee.
\vspace{-0.3cm}
  \subsection{Robustness guarantee with 1-NN retrieval}\label{sec:guarantee}
  \vspace{-0.1cm}
  Let $R_x$ be the subset of reference $R$ in which the samples have the same label as $x$, and let $R/R_x$ be its complement in $R$.
  We note that the 1-NN retrieval score of a sample $x$ depends only on its nearest embedding in $R_x$ and $R/R_x$. For arbitrary encoder $h$, if the distance between $h(x)$ and its nearest embedding in $R_x$ is smaller than the distance between $h(x)$ and its nearest embedding in $R/R_x$, the 1-NN retrieval score is 1; otherwise, the score is 0. To use this property, we have the following definition of minimum margin.
  \begin{definition}\label{def:d} (Minimum margin) 
  $$d(x;h)\hspace{-0.1cm}:=\hspace{-0.2cm}\min_{x_2\in R/R_x}\hspace{-0.1cm} ||h(x)-h(x_2)||_2 -  \min_{x_1\in R_x}\hspace{-0.1cm}||h(x)-h(x_1)||_2,$$
  where $h$ is arbitrary embedding model.
  \end{definition}
  If $d(x;h)>0$, the retrieval score of $x$ is 1 and otherwise 0. 
  In this work, we only consider the certified robustness of ``correctly-retrieved samples'', i.e., samples with retrieval score 1. Thus, in order to make the retrieval score of $x+\delta$ invariant to the perturbation $\delta$, we expect $d(x+\delta;h)>0$. 
  In the following theorem, we show an important property of $\delta$, with which the retrieval score is invariant to adversarial perturbation.
  \begin{lemma}\label{thm:1}
  For any embedding model $h$, if the retrieval score of $x$ w.r.t. $h$ is 1 and $$||h(x)-h(x+\delta)||_2<\frac{d(x;h)}{2},$$ the retrieval score of $x+\delta$ w.r.t. $h$ is also 1.
  \end{lemma}
  \begin{proof} Recall $R_x$ is the subset of the reference set $R$ in which the samples have the same label as $x$. With the classifier $h$, denote the nearest embedding of sample $x$ in $R_x$ by $x^+$, 
  it is not hard to observe
  $$d(x,h) = \min_{y\in R/R_x}||h(x)-h(y)||_2 -  ||h(x)-h(x^+)||_2,$$
  thus we have
  $$||h(x)-h(x^+)||_2\leq ||h(x)-h(y)||_2 -  d(x,h), \forall y\in R/R_x.$$
  For an arbitrary $y\in R/R_x$, we will show $||h(x+\delta)-h(x^+)||_2 <||h(x+\delta)-h(y)||_2$, if $||h(x)-h(x+\delta)||_2<\frac{d(x;h)}{2}.$
  \begin{align*}
      &||h(x+\delta)-h(x^+)||_2\\&\leq ||h(x+\delta)-h(x)||_2 + ||h(x)-h(x^+)||_2\\
      &< \frac{d}{2}+ ||h(x)-h(x^+)||_2\\
      &\leq \frac{d}{2}+ ||h(x)-h(y)||_2 -  d\\
      &\leq ||h(x)-h(y)||_2 -  \frac{d}{2}\\
      &< ||h(x)-h(y)||_2 - ||h(x+\delta)-h(x)||_2\\
      &\leq ||h(x+\delta)-h(y)||_2
  \end{align*}
  The nearest embedding of $x+\delta$ might change (not $x^+$), but still have the same label as $x$, which means the 1-NN retrieval score $x+\delta$ is still 1.
  \end{proof}
 
  If $h$ is a Lipschitz continuous embedding, i.e., there exists a constant $L$ such that $$||h(x)-h(y)||_2\leq L||x-y||_2,$$ 
  we can choose perturbation $\delta$ such that its $\ell_2$ norm is bounded by $\frac{d(x;h)}{2L}$. In this case,  $$||h(x)-h(x+\delta)||_2\leq L||x-(x+\delta)||_2<\frac{d(x;h)}{2}.$$
  That is, $h$ is guaranteed to be robust against any perturbation as long as its $\ell_2$ radius is bounded by $\frac{d(x;h)}{2L}$. Our following subsections will focus on designing an embedding model with a bounded Lipschitz constant.
 
\vspace{-0.3cm}
\subsection{Building a Lipschitz continuous model}
\label{sec:lipschitz}
\vspace{-0.1cm}
It is commonly known that deep neural networks are not Lipschitz continuous~\cite{weng2018evaluating}. To build a Lipschitz continuous embedding, one approach is by using Lipschitz preserving layers, e.g., orthogonal convolution neural networks (OCNN) \cite{wang2020orthogonal}. However, OCNN is hard to be trained due to its strong constraints imposed on each layer.
Another approach is by applying randomized smoothing \cite{cohen2019certified} to the base embedding model. It has been shown that in the classification tasks, the smoothed classifier is $\ell_2$-Lipschitz bounded \cite{salman2019provably}. In this work, we show that we can build a smoothed embedding model which also has bounded Lipschitz constant beyond the classification problem. In particular,
given a base embedding model $h:\mathbb{R}^d \to \mathbb{R}^k$, a sample $x$ and a distribution $q$, the smoothed embedding $g$ is given by 
  $$g(x) = \mathbb{E}_{z\sim q}[h(x+z)].$$
  Different from the randomized smoothing method in the classification task, where $h$ and $g$ represent the probability of correct predictions, in the embedding models $h(x)$ and $g(x)$ represent the feature vectors of sample $x$. In the next theorem, we show that if we select $q$ as a Gaussian distribution, the smoothed model $g$ has a bounded Lipschitz constant.
  \begin{lemma}\label{thm:2}
  If $q\sim\mathcal{N}(0,\sigma^2I)$, for arbitrary samples $x,y,$ \begin{equation}\label{eqn:actual}
  \begin{split}
  &||g(x)-g(y)||_2\\
  &\leq 2F\left(\Phi\left(\frac{||x-y||_2}{2\sigma}\right)-\Phi\left(\frac{-||x-y||_2}{2\sigma}\right)\right),
  \end{split}
  \end{equation}
  where $\Phi$ is the cumulative density function of $\mathcal{N}(0,1)$ and $F$ is the maximum $\ell_2$ norm of the base embedding model $h$.
  \end{lemma}
  Detailed proof is in \autoref{sec:proof3.3}. 

  \noindent\textbf{Tightness of this bound.} Consider a one-dimensional dataset $X$ in $\mathbb{R}$ and an embedding model $h(x) = Fsign(x)$. The smoothed model $g(x) = \mathbb{E}_{z\sim \mathcal{N}(0,\sigma^2) }[h(x+z)] = F(\Phi(\frac{x}{\sigma})-\Phi(-\frac{x}{\sigma}))$. Given two samples $x$ and $-x$, the difference between $g(x)$ and $g(-x)$ is $2F(\Phi(\frac{x}{\sigma})-\Phi(-\frac{x}{\sigma}))$, which reaches the upper bound in \autoref{eqn:actual}.
As $\Phi(z)-\Phi(-z)\leq \sqrt{\frac{2}{\pi}}z$ for $z\geq0$, we have 
\begin{equation*}
\begin{split}
  &2F\left(\Phi\left(\frac{||x-y||_2}{2\sigma}\right)-\Phi\left(\frac{-||x-y||_2}{2\sigma}\right)\right)\\
  &\leq F\sqrt{\frac{2}{\pi\sigma^2}}||x-y||_2.
\end{split}
\end{equation*}
  Thus the smoothed model $g$ has bounded Lipschitz constant, and bounded $\ell_2$ perturbation on the input will result in bounded shift of its embedding.
 
\begin{figure}[t]
    \centering
    \includegraphics[height=5cm]{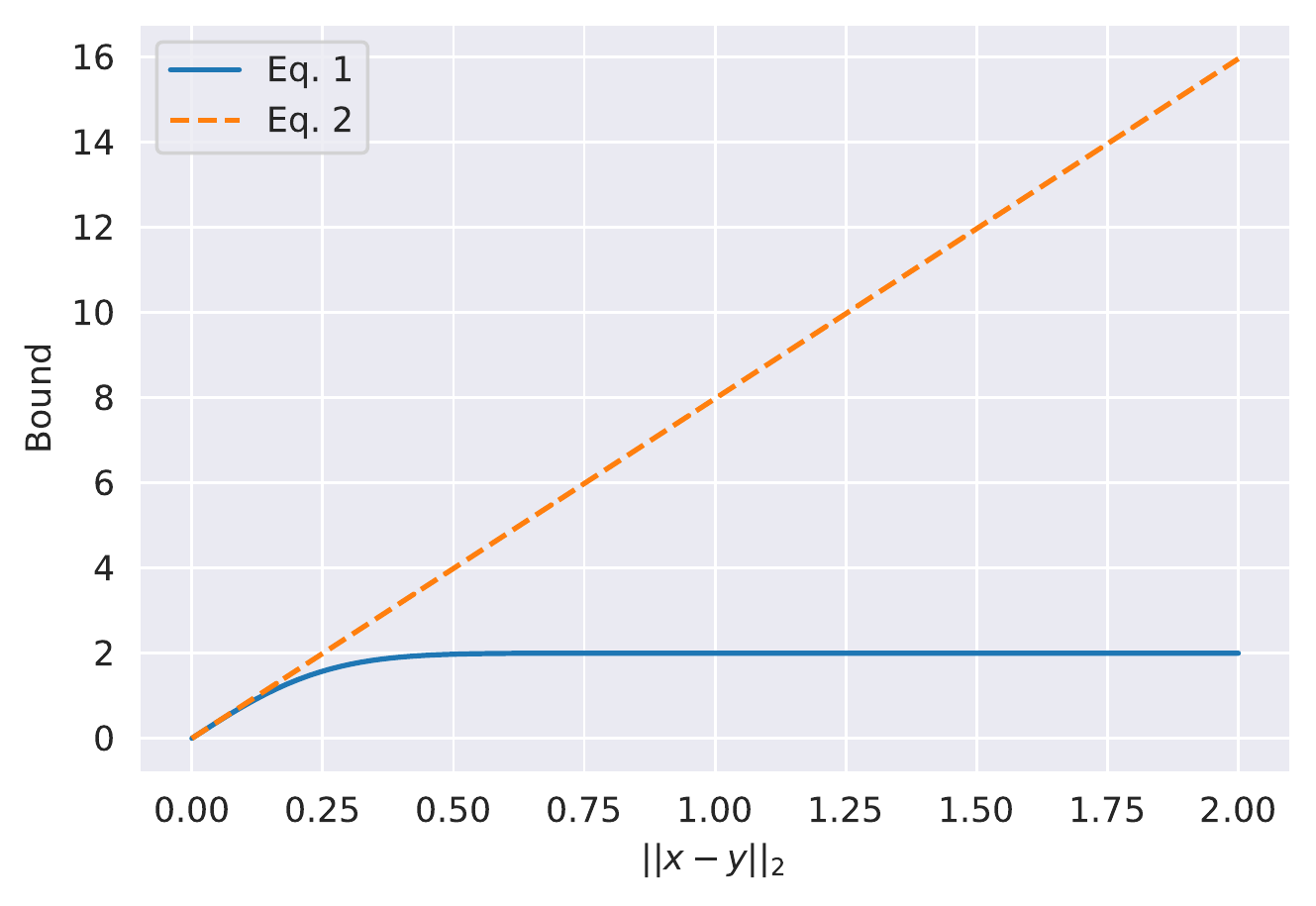}
    \vspace{-0.3cm}
    \caption{Comparison of the upper bounds of $||g(x)-g(y)||_2$ given by \autoref{eqn:actual} and \autoref{eqn:estlip}, where we set $F=1,\sigma = 0.1$.}
    \label{fig:compare}
\end{figure}
\subsection{Calculating certified radius}
\begin{definition} (Certified radius)
Given a sample $x$ and an embedding model $h$, the certified radius $r(x;h)$ is the radius of the largest $\ell_2$ ball, such that all perturbations $\delta$ within the ball cannot change the retrieval score of the sample $x$: 
$$r(x;h):=\max_{r\geq0} r, \ \textup{s.t.}\  R_1(x)=R_1(x+\delta), \forall \delta\hspace{-0.1cm}\in\hspace{-0.1cm}\{||\delta||_2<r\},$$
where $R_1(x)$ is the retrieval score of $x$.
\end{definition}
Following the discussion in  \autoref{sec:lipschitz}, the smoothed embedding model $g$ satisfies \begin{equation}\label{eqn:estlip}
||g(x)-g(y)||_2 \leq F\sqrt{\frac{2}{\pi\sigma^2}}||x-y||_2.\end{equation} Thus if we choose $r(x;g) = \frac{\sigma\sqrt{\pi}}{2\sqrt{2}F}d(x;g),$ for all $\delta$s with $||\delta||_2<r(x;g)$, we have $||g(x)-g(x+\delta)||_2<\frac{d(x;g)}{2}$.

However, \autoref{eqn:estlip} is looser than \autoref{eqn:actual}, and the certified radius computed by \autoref{eqn:estlip} is smaller than the radius given by \autoref{eqn:actual}. As shown in \autoref{fig:compare}, when $F=1$ and $\sigma = 0.1$, the Lipschitz bound of \autoref{eqn:estlip} is much worse than that of \autoref{eqn:actual} when $||x-y||_2$ is moderately large. Thus we will use the tighter bound (\autoref{eqn:actual}) to calculate our certified radius.

\begin{theorem}
\label{prop:1}
For any sample $x$ and the smoothed embedding $g$, with \autoref{eqn:actual}, if $d(x;g)>0$, the certified radius of $x$ is
\begin{equation}\label{eqn:radius}
    r(x;g) =  2\sigma\Phi^{-1}\left(\frac{1}{2}+\frac{d(x;g)}{8F}\right).
\end{equation}
\end{theorem}
 \begin{proof}
  From Lemma \ref{thm:1}, we know that the 1-NN retrieval score of a sample $x$ with smoothed embedding model $g$ does not change when $$||g(x)-g(x+\delta)||_2<\frac{d(x;g)}{2},$$
  From Lemma \ref{thm:2}, we have 
  $$||g(x)-g(x+\delta)||_2 \leq 2F(\Phi(\frac{||\delta||_2}{2\sigma})-\Phi(\frac{-||\delta||_2}{2\sigma})),$$
  thus if $\delta$ satisfies $$2F(\Phi(\frac{||\delta||_2}{2\sigma})-\Phi(\frac{-||\delta||_2}{2\sigma}))<\frac{d(x;g)}{2},$$ the 1-NN retrieval score of $x$ will not change. As $$2F(\Phi(\frac{||\delta||_2}{2\sigma})-\Phi(\frac{-||\delta||_2}{2\sigma})) = 2F(2\Phi(\frac{||\delta||_2}{2\sigma})-1),$$ by solving 
  $$2F(2\Phi(\frac{||\delta||_2}{2\sigma})-1)<\frac{d(x;g)}{2},$$
  we have $||\delta||_2<2\sigma\Phi^{-1}(\frac{1}{2}+\frac{d(x;g)}{8F})$, thus the certified radius $r(x,g) = 2\sigma\Phi^{-1}(\frac{1}{2}+\frac{d(x;g)}{8F})$
  \end{proof}
In practice, it is hard to compute the smoothed model $g = \mathbb{E}_{z\sim \mathcal{N}(0,\sigma^2I_d) }[h(x+z)]$ and the minimum margin $d(x;g)$ in a closed form. To resolve the issue, we use Monte-Carlo sampling to estimate $g(x)$ and calculate a probabilistic lower bound of $d(x;g)$. Denote the Monte-Carlo estimation of $g(x)$ by $$\hat{g}(x):=\frac{1}{n}\sum_{i=1}^{n}h(x+z_i),$$ where $\{z_1,...,z_n\}$ are sampled i.i.d. from $\mathcal{N}(0,\sigma^2I_d)$. By matrix Chernoff bound \cite{ahlswede2002strong,tropp2012user}, we have the following theorem.
\begin{lemma} \label{thm:3}
With $g(x)\in \mathbb{R}^k$ and $\hat{g}(x):=\frac{1}{n}\sum_{i=1}^{n}h(x+z_i)$, where $\{z_1,...,z_n\}$ are the Monte-Carlo samples of $\mathcal{N}(0,\sigma^2I_d)$, we have
$$\mathbb{P}(||g(x)-\hat{g}(x)||_2>\epsilon)\leq(k+1)\exp\left(-\frac{3\epsilon^2n}{8F^2}\right).$$
\end{lemma}
  \begin{proof}
 \noindent We start with an introduction of the matrix Chernoff bound.
  \begin{lemma}\label{lm:1}(Matrix Chernoff bound \cite{ahlswede2002strong,tropp2012user})
  Let $M_1, ..., M_t$ be independent matrix valued random variables such that $ M_{i}\in \mathbb {C} ^{d_{1}\times d_{2}}$ and $ \mathbb {E} [M_{i}]=\mu$. Denote the operator norm of the matrix $M$ by $\lVert M\rVert $. If $\lVert M_{i}\rVert \leq \gamma$ holds almost surely for all $i\in \{1,\ldots ,t\}$, then for every $\epsilon > 0$
  \begin{equation}
       \mathbb{P}\left(\left\|{\frac {1}{t}}\sum _{i=1}^{t}M_{i}-\mu\right\|>\varepsilon \right)\leq (d_{1}+d_{2})\exp \left(-{\frac {3\varepsilon ^{2}t}{8\gamma ^{2}}}\right).
  \end{equation}
  \end{lemma}
\noindent  Now we prove Lemma \ref{thm:3}.

  Assume we have $n$ i.i.d. random variables $G_1,...,G_n$, each $G_i$ have the same distribution with $f(x+Z)$, where $Z$ is an arbitrary smoothing distribution. Notice in our paper $Z\sim\mathcal{N}(0,\sigma^2I)$, but this does not influence the conclusion. We will show for any $Z$, the bound in Lemma \ref{thm:3} with $g(x) = \mathbb{E}[f(x+Z)]$ always hold. 
  
  Since $G_i$ has the same distribution with $f(x+Z)$, we have $G_i\in\mathbb{R}^{k\times 1}$ and $\mathbb{E}[G_i] = g(x)$. The $\ell_2$ operator norm of $G_i$ is given by,
  \begin{equation*}
      \begin{split}
          ||G_i|| &= \sup_{||v||_2\leq 1,v\in\mathbb{R}^{k\times 1}}<G_i,v>\\
          &=||G_i||_2\leq \sup_{x}||f(x)||_2 = F.
      \end{split}
  \end{equation*}
  Thus with Lemma \ref{lm:1} we have
  \begin{equation}
  \mathbb{P}\left(\left\|{\frac {1}{n}}\sum _{i=1}^{n}G_i-g(x)\right\|>\epsilon \right)\leq (k+1)\exp \left(-{\frac {3\epsilon ^{2}n}{8F ^{2}}}\right).
  \end{equation}
  As $\hat{g}(x) = \sum_{i=1}^{n}f(x+Z_i)$, where $Z_i$ are sampled independently from $Z$, $\hat{g}(x)$ follows the distribution of $\frac{1}{n}\sum _{i=1}^{n}G_i$. Therefore we have 
  $$\mathbb{P}(||g(x)-\hat{g}(x)||_2>\epsilon)\leq(k+1)\exp(-\frac{3\epsilon^2n}{8F^2})$$
  \end{proof}

Taking $\alpha = (k+1)\exp(-\frac{3\epsilon^2n}{8F^2})$, we have that with probability at least $1-\alpha$, the $\ell_2$ norm of $g(x)-\hat{g}(x)$ is upper bounded by $\sqrt{8F^2\ln(\frac{k+1}{\alpha})/{3n}}$. Thus the error of the estimation is asymptotically decreasing with rate $O(1/\sqrt{n})$, and we can obtain arbitrarily accurate estimation of $g(x)$ with large Monte-Carlo samplings.
\begin{lemma}
\label{prop:2}
With probability at least $1-\alpha$,
$${d}(x;g)\geq d(x;\hat{g})-4\sqrt{\frac{8F^2\ln\left(\frac{k+1}{\alpha/4}\right)}{3n}}=:\underline{d}(x;g).$$
\end{lemma}

Detailed proof is in \autoref{sec:proof3.8}.

With the lower bound estimation and Theorem \ref{prop:1}, we are able to calculate the certified radius for any given sample $x$.
\begin{proposition}\label{prop:3}(Monte-Carlo calculation of certified radius)
If $\underline{d}(x;g)>0$, with probability at least $1-\alpha$,
\begin{equation}\label{eqn:radiusreal}
    r(x;g) \geq  2\sigma\Phi^{-1}\left(\frac{1}{2}+\frac{\underline{d}(x;g)}{8F}\right).
\end{equation}
\end{proposition}
  \begin{proof}
This proposition is a naive combination of Theorem \ref{prop:1} and Lemma \ref{prop:2}, as $\Phi^{-1}(x)$ is a monotonously increasing function with $x$ and ${d}(x,g)\geq\underline{d}(x,g)$ with probability $1-\alpha$, obviously 
\begin{equation*}
\begin{aligned}
r(x;g) &=  2\sigma\Phi^{-1}\left(\frac{1}{2}+\frac{d(x;g)}{8F}\right)\\
&\geq 2\sigma\Phi^{-1}\left(\frac{1}{2}+\frac{\underline{d}(x;g)}{8F}\right)
\end{aligned}
\end{equation*}
with probability $1-\alpha$.
\end{proof}

\begin{remark}
\label{remark:dimension}
Our Monte-Carlo calculation of certified radius logarithmically depends on the dimension of the embedding space $k$. So our method works even when the dimension of the embedding space is high. This is different from randomized smoothing as its output is required to be a one-dimensional scalar.
\end{remark}
\begin{algorithm}[t]
	\SetAlgoLined
	\KwIn{training set $X$; number of random samples $n$; base embedding model $h$; standard derivation of Gaussian $\sigma$; confidence level $\alpha$.}
	
	Initialize class balanced sampler $S$;\\
	\For{$x\in X$}{
	sample $N$ random variable $z_1,...,z_N$ from $\mathcal{N}(0,\sigma^2I)$;\\
	calculate $\hat{g}(x) = \frac{1}{n}\sum_{i=1}^{n}h(x+z_i)$;\\
	}
	\For{$x\in X$}{
	    calculate $\underline{d}(x;g)$ by Lemma \ref{prop:2};\\
	    \eIf{$\underline{d}(x;g)<0$}{
	    $r(x)=-1$; (reject this sample because its retrieval score is 0)}
	    {
	     calculate certified radius $r(x) = 2\sigma\Phi^{-1}(1/2+\underline{d}(x;g)/8F)$;}
	}
	\textbf{return} all certified radius $r$.
	\caption{Certified Radius by RetrievalGuard}
	\label{alg:1}
\end{algorithm}
Algorithm \ref{alg:1} describes our procedure of calculating the certified radius. We want to emphasize that the base embedding model $h$ does not have any robustness guarantee; only its smoothed version $g$ is certifiably robust. 

\medskip
\noindent\textbf{New techniques compared to randomized smoothing.}
a) Randomized smoothing is designed for the classification tasks, where one can prove that the prediction score w.r.t. the smoothed classifier is larger than 0.5 if the true label is 1. However, in the 1-NN retrieval tasks, we need to identify the conditions under which the 1-NN search is robust against perturbation attacks (Lemma \ref{thm:1}). b) We prove the Lipschitzness of the smoothed embedding model in Lemma \ref{thm:2}. Compared to randomized smoothing where the output is a one-dimensional scalar, in the image retrieval tasks we need to take into account the high-dimensional nature of the embedding space (see Remark \ref{remark:dimension}). c) Compared to the probabilistic guarantee in randomized smoothing, which is a result of Neyman-Pearson lemma, the probabilistic guarantee (Lemma \ref{prop:2}) for our model is based on brand new analysis of minimum margin in Definition \ref{def:d}. The minimum margin $d(x,g)$ depends on multiple samples, and we need to provide a union bound to characterize the uncertainty of all relevant samples.



\begin{figure}[t]
		\begin{minipage}[t]{1\linewidth}
			\centering
			\includegraphics[height=5cm]{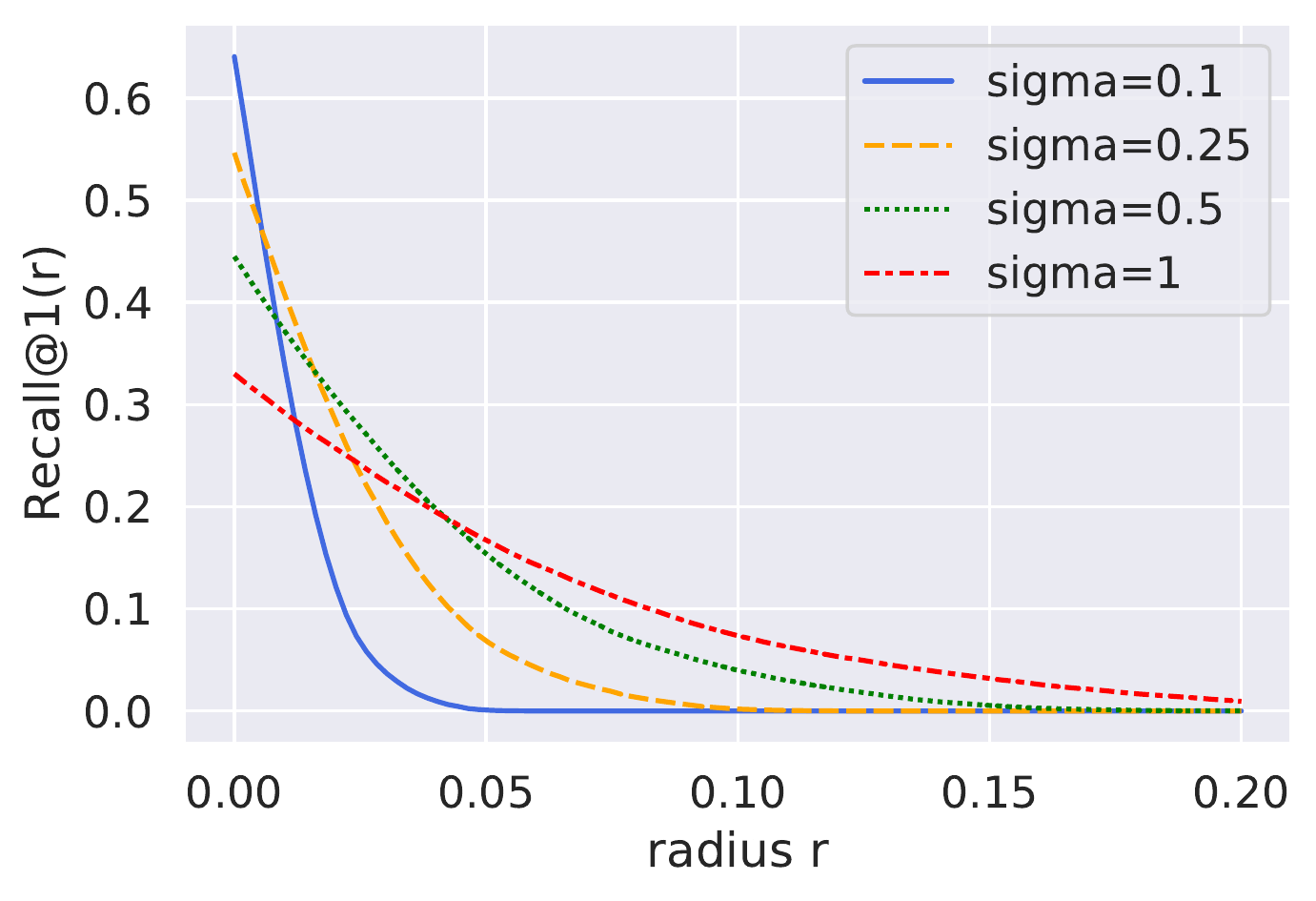}	
			\vspace{-6pt}
			\label{fig:margin_car}
		\end{minipage}
\\
		\begin{minipage}[t]{1\linewidth}
			\centering
			\includegraphics[height=5cm]{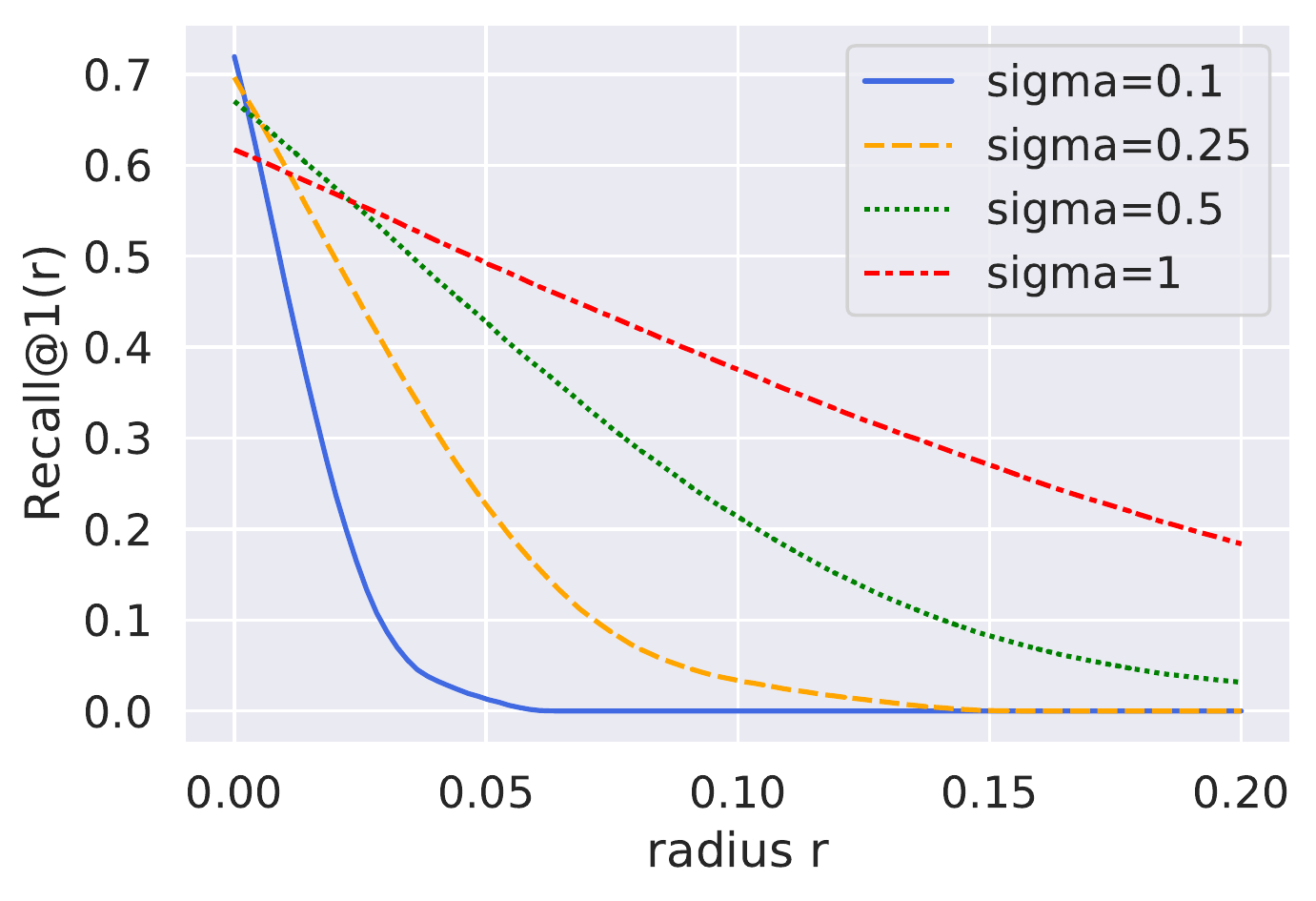}	
			\vspace{-6pt}
			\label{fig:margin_car}
		\end{minipage}
        \caption{Experiments with DML+RetrievalGuard on image retrieval benchmarks with different $\sigma$. \textbf{Top}: DML+RetrievalGuard on Online-Products. \textbf{Bottem}: GDML+RetrievalGuard on Online-Products.}
        \label{fig:DML}
\end{figure}

\section{Experiments}
In the experiments, we use the metric learning model with margin loss \cite{wu2017sampling} as our base embedding model. We apply our RetrievalGuard approach on vanilla metric learning (DML) and the DML augmented by Gaussian noise (GDML) to build the smoothed DML and compare them on three benchmarks. We emphasize that all results reported in this section are from the smoothed models $g$ instead of the base models $h$, as we can only provide robustness guarantee for $g$.
\subsection{Deep metric learning with margin loss}
Margin loss is a tuple-based metric loss, which requires (anchor, positive, negative) triplets as input. The anchor and the positive point are expected to be in the same class while the anchor and the negative point should be in the different classes. Denote the (anchor, positive, negative) by $(x,x^+,x^-)$ and the distribution of the triplets by $p_{tri}$. The margin loss \cite{wu2017sampling} is defined as
\begin{align*}
    L(h,\beta) = &\mathbb{E}_{(x,x^+,x^-)\sim p_{tri}}[\\
    &(||h(x)-h(x^+)||_2-\beta+\gamma)_+\\
    +&(\beta-||h(x)-h(x^-)||_2+\gamma)_+],
\end{align*}
where $\beta$ is a learnable parameter with initial value 0.6 or 1.2 and learning rate $0.0005$. $\gamma=0.2$ is a fixed triplet margin. In the margin loss, $p_{tri}$ is given by a distance sampling method, such that the probability of sampling a negative point with large distance to $x$ in the embedding space is much larger than that of sampling a negative point with small distance to $x$.

\subsection{Gaussian augmented model}\label{sec:gdml}
The norm of Gaussian noise sampled from $\mathcal{N}(0,\sigma^2I_d)$ is of magnitude $\Theta(\sigma\sqrt{d})$ with high probability \cite{zhang2020black}. With moderately large $\sigma$, the distribution of natural images has nearly disjoint support from the distribution of Gaussian-perturbed images. It is therefore hard for the base model $h$ to generate effective embedding, if it can only get access to natural images. As a result, the smoothed model $g$, which is estimated by averaging the base embedding, may suffer from poor performance. A solution to resolve this issue is by training the base embedding model $h$ with Gaussian augmented images \cite{cohen2019certified}. \autoref{fig:DML} shows that using Gaussian augmented model as the base model outperforms using the vanilla model as the base model in all settings. The objective function for the Gaussian augmented model is
\begin{align*}
    L(h,\beta)\hspace{-0.05cm}=&\mathbb{E}_{(x,x^+,x^-)\sim p_{tri}}[\\
    &(||h(x+Z_1)-h(x^++Z_2)||_2-\beta+\gamma)_+\\
    +&(\beta-||h(x+Z_1)-h(x^-+Z_3)||_2+\gamma)_+],
\end{align*}
where $Z_1,Z_2,$ and $Z_3$ are Gaussian random variables and $(t)_+=\max\{t,0\}$.

\begin{figure*}[ht]
		\begin{minipage}[t]{.32\linewidth}
			\centering
			\includegraphics[height=4cm]{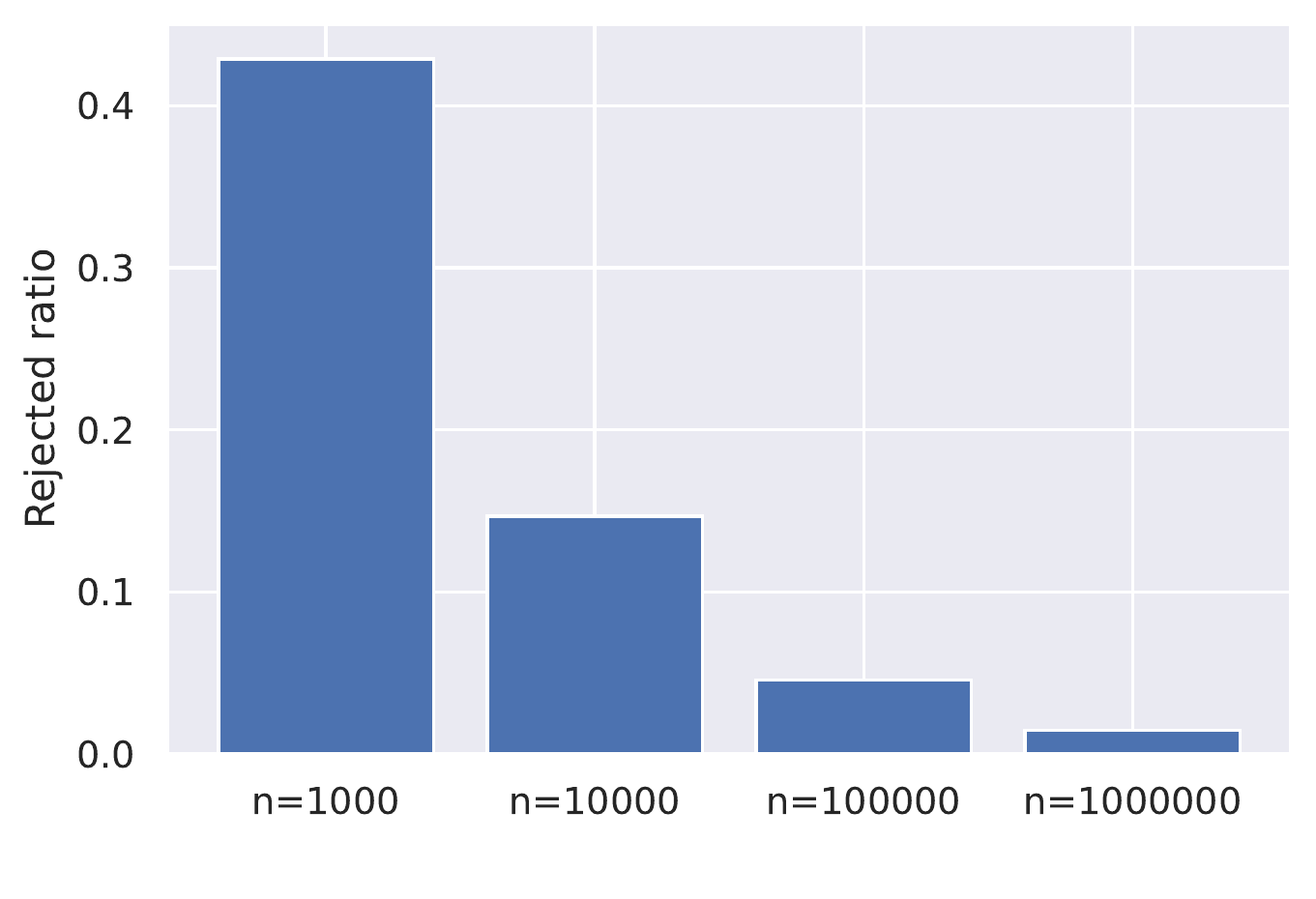}	
			\vspace{-6pt}
			\label{fig:margin_car}
		\end{minipage}
		\begin{minipage}[t]{.32\linewidth}
			\centering
			\includegraphics[height=4cm]{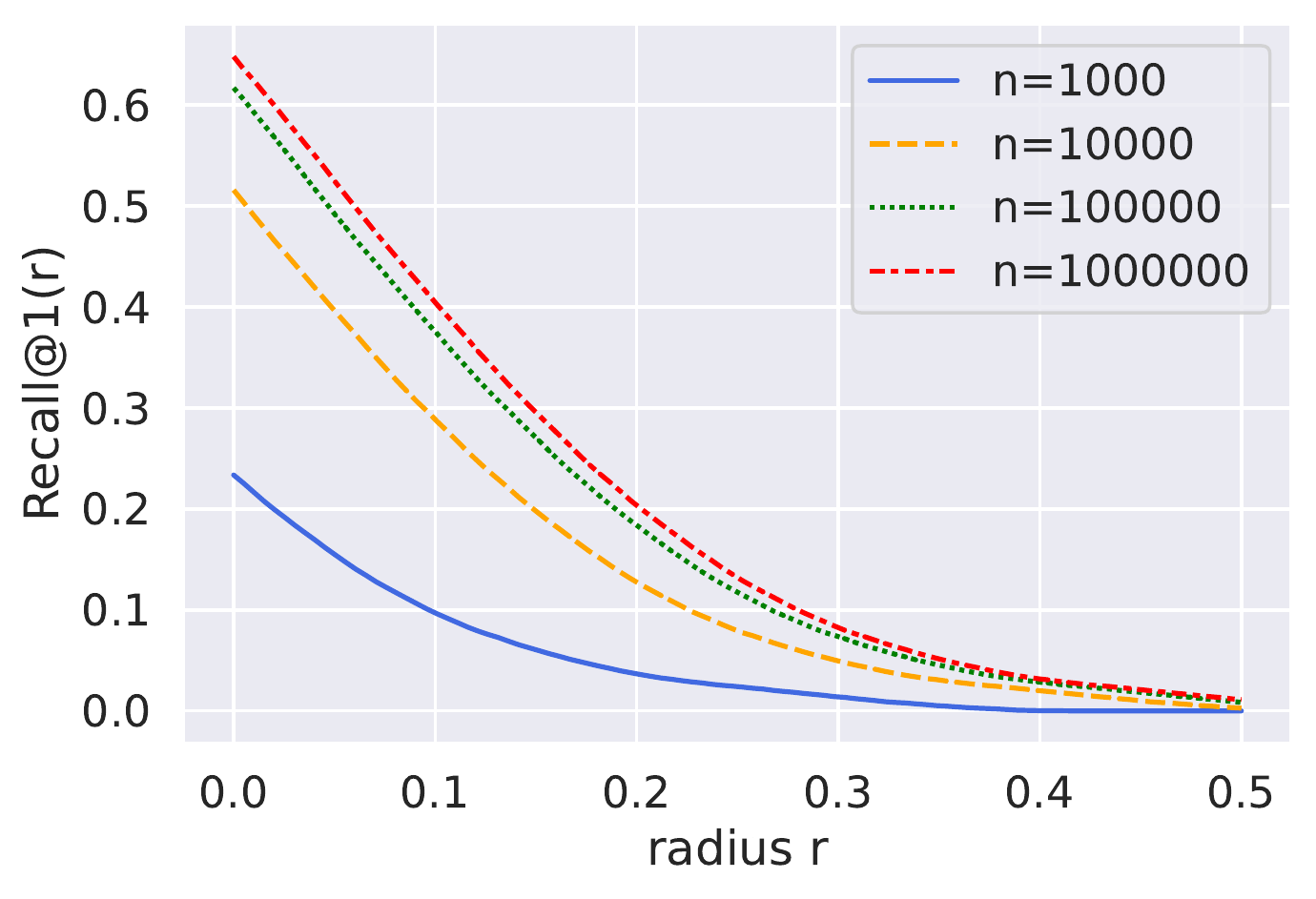}	
			\vspace{-6pt}
			\label{fig:ms_cub}
		\end{minipage}
		\begin{minipage}[t]{.32\linewidth}
			\centering
			\includegraphics[height=4cm]{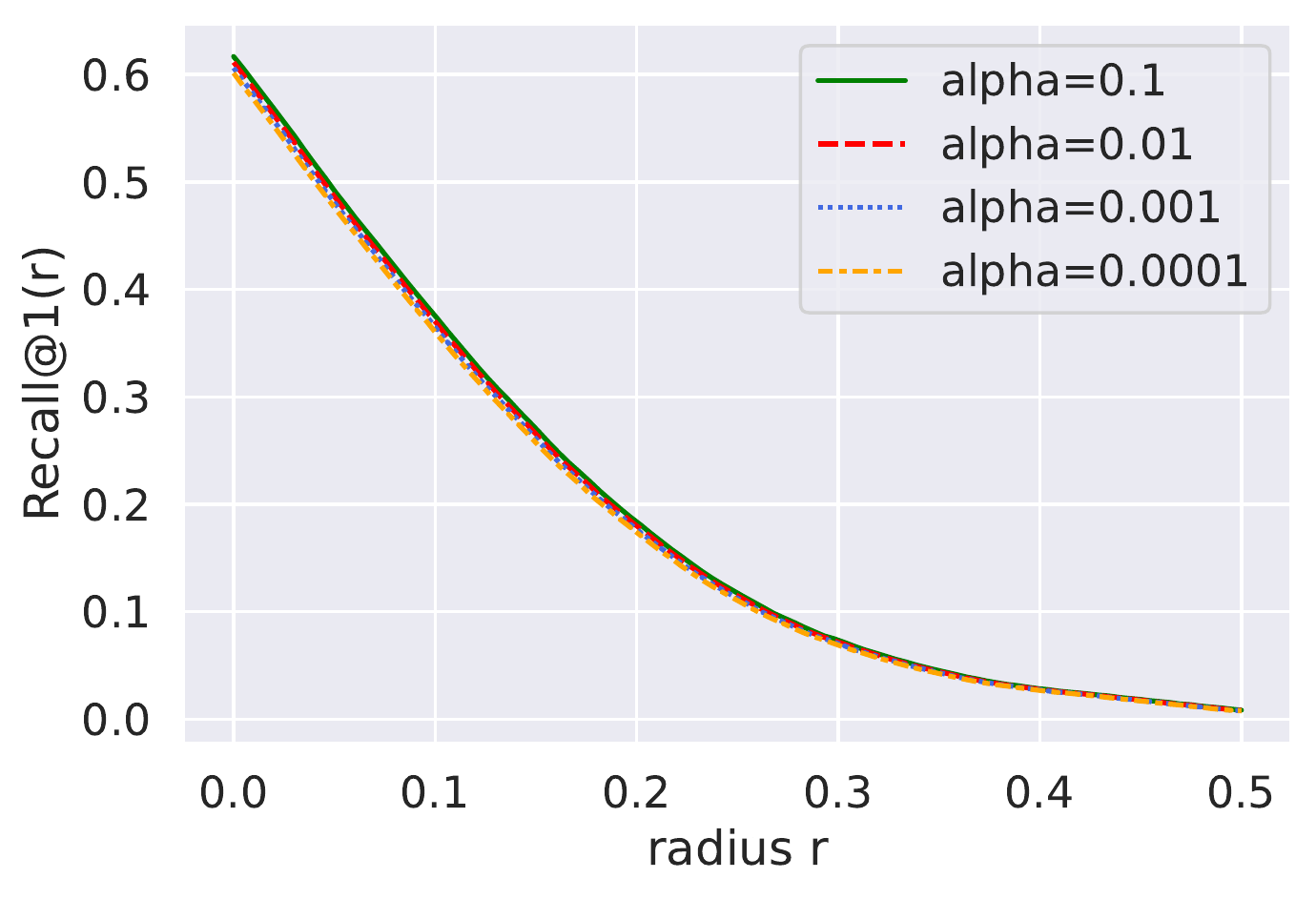}	
			\vspace{-6pt}
			\label{fig:margin_car}
		\end{minipage}
	\vspace{-10pt}

\caption{Experiments with GDML+RetrievalGuard on Online-Products with $\sigma=1$. \textbf{Left}: Comparison of rejected ratio with the number of Monte-Carlo samples $n$. The rejected ratio is the ratio of samples with $d(x,\hat{g})>0$ but $\underline{d}(x,g)\leq0$. \textbf{Middle}: Comparison of $Recall@1(r)$ if the number
of Monte-Carlo samples $n$ is larger or smaller. \textbf{Right}: Comparison of $Recall@1(r)$ under varying values of $\alpha$.}
\vspace{-0.5cm}
\label{fig:ablation}

\end{figure*}
\begin{figure}[ht]
    \centering
    \includegraphics[height=5cm]{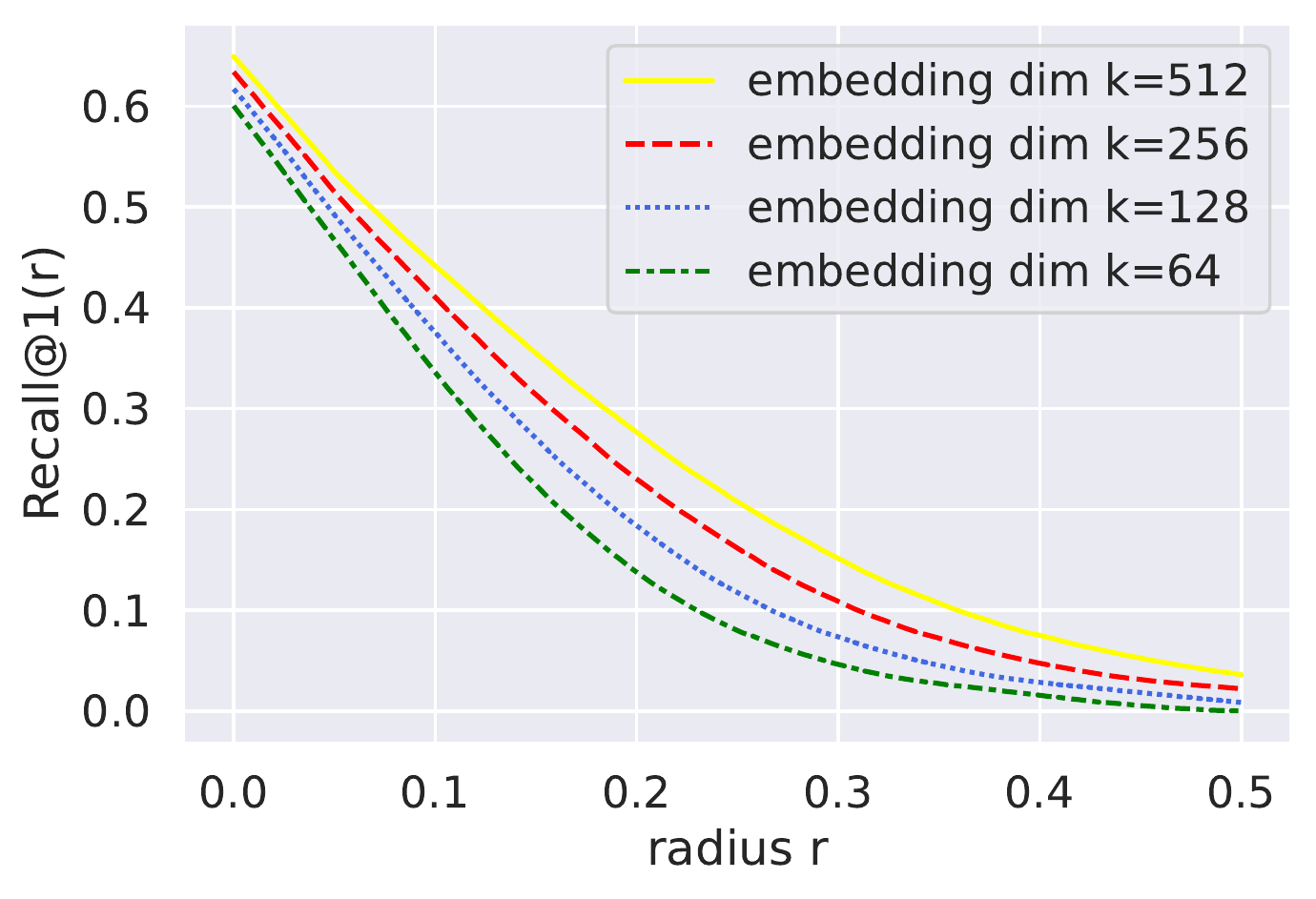}

    \caption{GDML+RetrievalGuard ($\sigma=1$) with different dimensions of embedding space on the Online-Products dataset.}
    \label{fig:ablationk}
	\vspace{-10pt}
\end{figure}
\subsection{Experimental settings}
\noindent\textbf{Datasets.} We run experiments with a popular dataset Online-Products of metric learning \cite{song2016deep}, which contains 120,053 product images. We use the first 11,318 classes of products as the training set and another 1,000 classes as the test set. The experiments with CUB200 \cite{WahCUB_200_2011} and CARS196 \cite{KrauseStarkDengFei-Fei_3DRR2013} are listed in \autoref{sec:additional experiments}

\noindent\textbf{Training hyper-parameters.}
We adapt the DML framework from \cite{roth2020revisiting} for our training. In all experiments, we use ResNet50 architecture \cite{he2016deep} pretrained on the ImageNet dataset \cite{krizhevsky2012imagenet} with frozen Batch-Normalization layers as our backbone. We first re-scale the images to $[0,1]^d$, then randomly resize and crop the images to $224 \times 224$
for training, and apply center crop to the same size for evaluation. The embedding dimension $k$ is 128 and the number of training epochs is 100. The learning rate is 1e-5 with multi-step learning rate scheduler 0.3 at the 30-th, 55-th, and 75-th epochs. We select the initial value of $\beta$ as 1.2, learning rate of $\beta$ as 0.0005, and $\gamma = 0.2$ in the margin loss. We also test the model performance under Gaussian noise with $\sigma = 0.1, 0.25, 0.5, 1$. As the embedding of metric learning models is $\ell_2$ normalized, we have $F=1$. For each sample, we generate 100,000 Monte-Carlo samples to estimate $g(x)$. The confidence level $\alpha$ is chosen as 0.01. The running time of RetrievalGuard for a single image evaluation with 100,000 Monte-Carlo samples on a 24GB Nvidia Tesla P40 GPU is about 3 minutes.
\noindent\textbf{Evaluation metrics.} We focus on the 1-NN retrieval task. A natural metric is the Recall@1 score, which is given by the average of 1-NN retrieval score of all samples, i.e.,
$Recall@1 = \frac{1}{N}\sum_{i=1}^{N}R_1(x_i).$
To evaluate the certified robustness of the 1-NN retrieval, we define
$Recall@1(r) = \frac{1}{N}\sum_{i=1}^{N}R_1(x_i)\mathbb{I}_{r(x_i,g)>r},$ which represents the averaged retrieval score of the samples such that the certified radius is larger than $r$. Note that $Recall@1 = Recall@1(0)$.


\subsection{Experimental results}

\subsubsection{Performance of the smoothed DML models.}

\autoref{fig:DML} shows the plot of $Recall@1(r)$ of DML+RG and GDML+RG models under varying values of $\sigma$.
We see that there is a robustness/accuracy trade-off~\cite{yang2020closer,zhang2019theoretically} controlled by $\sigma$. When $\sigma$ is low, small radii can be certified with high retrieval score, while large radii cannot be certified. When $\sigma$ is high, large radii can be certified, while small radii are certified with a low retrieval score. Besides, the GDML+RG models outperform the DML+RG models in all experiments, which is consistent with our discussions in \autoref{sec:gdml}.
\vspace{-0.1cm}
\subsubsection{Ablation study}
\vspace{-0.1cm}
We study the effect of number of Monte-Carlo samples $n$, the failure probability $\alpha$, and the embedding size $k$ on the model robustness. All experiments are run on the GDML+RetrievalGuard model with $\sigma = 1$ and the Online-Products dataset.

\autoref{fig:ablation} (\textbf{left}) plots the rejected ratio under different numbers of Monte-Carlo samples $n$. We use a fixed $\alpha = 0.01$ in our experiment. The rejected ratio is the ratio of samples with retrieval score 1 on the estimated model $\hat{g}$, i.e., $d(x;\hat{g})>0$, but $x$ is rejected as $\underline{d}(x;g)\leq0$. The ratio of rejected samples is decreasing w.r.t. $n$. When $n=10,000$, there are 14\% samples being rejected. In our experiments, we choose $n = 100,000$ with roughly 4\% rejected ratio.

\autoref{fig:ablation} (\textbf{middle}) illustrates the $Recall@1(r)$ under different numbers of Monte-Carlo samples $n$. We use a fixed $\alpha = 0.01$ in the experiment. When $n$ is increased from $1,000$ to $100,000$, the improvement of $Recall@1(r)$ is significant. When $n$ is increased from $100,000$ to $1,000,000$, the improvement of $Recall@1(r)$ is relatively small. Therefore, we choose $n = 100,000$ in our experiments.

\autoref{fig:ablation} (\textbf{right}) draws the $Recall@1(r)$ under different confidence levels $\alpha$. We fix $n = 100,000$ in the experiment. It shows that $Recall@1(r)$ is stable under varying values of $\alpha$, which indicates that our robustness guarantee is not sensitive to $\alpha$.

\autoref{fig:ablationk} shows the $Recall@1(r)$ under different dimensions $k$ of embedding space. We use a fixed $n = 100,000$ and $\alpha = 0.01$ in the experiment. It shows that the models are more robust with larger $k$. This is because high-dimensional embedding space can improve the expressive power of the DML models, and dissimilar samples can be separated with a large margin, i.e., $\underline{d}(x,g)$ is large.

\section{Conclusion}

In this work, we propose RetrievalGuard, the first provably robust 1-NN image retrieval model, by smoothing the vanilla embedding model with a Gaussian distribution. 
We prove that, with arbitrary perturbation $\delta$, whose $\ell_2$ norm is bounded by the certified radius, the 1-NN retrieval score of the perturbed samples on the smoothed model does not change. We empirically demonstrate the effectiveness of our model on image retrieval tasks with Online-Products. Future works include designing $\ell_p$ certificate algorithms and extending our algorithm to $k$-NN image retrieval tasks.

\section*{Acknowledgement}
Hongyang Zhang was supported in part by an NSERC Discovery Grant. Yihan Wu and Heng Huang were partially supported by  NSF IIS 1845666, 1852606, 1838627, 1837956, 1956002, IIA 2040588.

\bibliography{example_paper}
\bibliographystyle{icml2022}

\newpage
\appendix
\onecolumn
  
\section{Proof of Lemma \ref{thm:2}}\label{sec:proof3.3}
  \begin{proof} Recall that $g(x):= \mathbb{E}_{z\sim\mathcal{N}(0,\sigma^2I_d)}[f(x+z)]=\int_{\mathbb{R}^d}f(x+z)p(z)dz$, where $p$ is the probability density function of $\mathcal{N}(0,\sigma^2I_d)$. The intuition of this proof is to seek the maximum discrepancy of the expectation $\mathbb{E}_{z\sim\mathcal{N}(0,\sigma^2I_d)}[f(x+z)]$ and $\mathbb{E}_{z\sim\mathcal{N}(0,\sigma^2I_d)}[f(y+z)]$ for arbitrary samples $x,y$.
  \begin{align*}
      ||g(x)-g(y)||_2 &= ||\int_{\mathbb{R}^d}f(x+z)p(z)dz-\int_{\mathbb{R}^d}f(y+z)p(z)dz||_2\\
      &=||\int_{\mathbb{R}^d}f(z)p(z-x)dz-\int_{\mathbb{R}^d}f(z)p(z-y)dz||_2\\
      &=||\int_{\mathbb{R}^d}f(z)(p(z-x)-p(z-y))dz||_2\\
      & = ||\int_{\mathbb{R}^d}f(z)(p(z-x)-p(z-y))_+dz+\int_{\mathbb{R}^d}f(z)(p(z-x)-p(z-y))_-dz||_2\\
      &\leq ||\int_{\mathbb{R}^d}f(z)(p(z-x)-p(z-y))_+dz||_2+||\int_{\mathbb{R}^d}f(z)(p(z-x)-p(z-y))_-dz||_2\\
      &\leq F(|\int_{\mathbb{R}^d}(p(z-x)-p(z-y))_+dz|+|\int_{\mathbb{R}^d}(p(z-x)-p(z-y))_-dz|)\\
      &=F(|\int_{\mathbb{R}^d}(p(z)-p(z+x-y))_+dz|+|\int_{\mathbb{R}^d}(p(z)-p(z+x-y))_-dz|)
  \end{align*}
  Now we need to calculate $\int_{\mathbb{R}^d}(p(z)-p(z+x-y))_+dz$ explicitly, as $p(z) = \frac{1}{(2\pi\sigma^2)^{d/2}}\exp(-\frac{z^Tz}{2\sigma})$
  \begin{align*}
  \int_{\mathbb{R}^d}(p(z)-p(z+x-y))_+dz &= \frac{1}{(2\pi\sigma^2)^{d/2}}\int_{\mathbb{R}^d}(\exp(-\frac{z^Tz}{2\sigma})-\exp(-\frac{(z+x-y)^T(z+x-y)}{2\sigma}))_+dz
  \end{align*}
  By solving $\exp(-\frac{z^Tz}{2\sigma})-\exp(-\frac{(z+x-y)^T(z+x-y)}{2\sigma})\geq 0$ we have $z\in\{z^T(x-y)\leq\frac{1}{2}(x-y)^T(x-y)\}:=D$.
  As Gaussian distribution is $\ell_2$ spherically symmetric, 
  we can make a unitary transformation such that $x-y$ located on the first axis in the $\mathbb{R}^d$ space, in this case $D = \{z_1\leq\frac{1}{2}||x-y||_2\}$ assuming w.l.o.g $x\geq y$, we have
  \begin{align*}
  \int_{\mathbb{R}^d}(p(z)-p(z+x-y))_+dz &= \frac{1}{(2\pi\sigma^2)^{d/2}}\int_{D}\exp(-\frac{z_1^2+\sum_{i=2}^dz_i^2}{2\sigma})-\exp(-\frac{(z_1+||x-y||_2)^2+\sum_{i=2}^dz_i^2}{2\sigma})dz\\
  &= \frac{1}{(2\pi\sigma^2)^{1/2}}\int_{z_1\leq \frac{||x-y||_2}{2}}\exp(-\frac{z_1^2}{2\sigma})-\exp(-\frac{(z_1+||x-y||_2)^2}{2\sigma})dz_1\\
  &=\Phi(\frac{||x-y||_2}{2\sigma})-\Phi(-\frac{||x-y||_2}{2\sigma})
  \end{align*}
  where $\Phi(x) = \frac{1}{(2\pi)^{1/2}}\int_{-\infty}^{x}\exp(-\frac{z^2}{2})dz$ is the cumulative density function of a standard normal distribution. Anogously we have 
  $$\int_{\mathbb{R}^d}(p(z)-p(z+x-y))_-dz = \Phi(-\frac{||x-y||_2}{2\sigma})-\Phi(\frac{||x-y||_2}{2\sigma})$$
  Thus 
  \begin{align*}
      ||g(x)-g(y)||_2 &=F(|\int_{\mathbb{R}^d}(p(z)-p(z+x-y))_+dz|+|\int_{\mathbb{R}^d}(p(z)-p(z+x-y))_-dz|)\\
      &=2F(\Phi(\frac{||x-y||_2}{2\sigma})-\Phi(-\frac{||x-y||_2}{2\sigma}))\\
  \end{align*}
  \end{proof}

  \section{Proof of Lemma \ref{prop:2}}\label{sec:proof3.8}
  \begin{proof} Recall $R_x$ is the subset of the reference set $R$ in which the samples have the same label as $x$. With the estimated classifier $\hat{g}$, denote the nearest embedding of sample $x$ in $R_x$ by $x^+$, and the nearest embedding of sample $x$ in $R_x$ by $x^-$, we have 
  $$d(x;\hat{g}) = ||\hat{g}(x)-\hat{g}(x^-)||_2-||\hat{g}(x)-\hat{g}(x^+)||_2$$
  According to \autoref{thm:3}, 
$$\mathbb{P}(||g(x)-\hat{g}(x)||_2>\epsilon)\leq(k+1)\exp(-\frac{3\epsilon^2n}{8F^2}),\forall x\in\mathbb{R}^d$$
  thus with probability at least $1-\alpha$
  $$||g(x)-\hat{g}(x)||_2\leq \sqrt{8F^2\ln(\frac{k+1}{\alpha})/{3n}},\forall x\in\mathbb{R}^d$$
  So we can bounded the difference between $||g(x)-g(y)||_2$ and $||\hat{g}(x)-\hat{g}(y)||_2$ for arbitrary sample $x$ and $y$ by
  \begin{equation}\label{eq:diff}
  \begin{aligned}
      |||g(x)-g(y)||_2-||\hat{g}(x)-\hat{g}(y)||_2|&\leq ||g(x)-g(y)-\hat{g}(x)+\hat{g}(y)||_2\\
      &\leq ||g(x)-\hat{g}(x)||_2+||g(y)-\hat{g}(y)||_2\\
      & = 2\sqrt{8F^2\ln(\frac{k+1}{\alpha})/{3n}}
  \end{aligned}
  \end{equation}
  with probability at least $1-2\alpha$. We now consider $$d(x;g):=\min_{x_2\in R/R_x} ||g(x)-g(x_2)||_2 -  \min_{x_1\in R_x}||g(x)-g(x_1)||_2$$
  Based on \autoref{eq:diff} we have
  $$\min_{x_2\in R/R_x} ||g(x)-g(x_2)||_2\geq ||\hat{g}(x)-\hat{g}(x^*)||_2-2\sqrt{8F^2\ln(\frac{k+1}{\alpha})/{3n}}\geq ||\hat{g}(x)-\hat{g}(x^-)||_2-2\sqrt{8F^2\ln(\frac{k+1}{\alpha})/{3n}}$$
  with probability at least $1-2\alpha$, where $x^*:=arg\min_{x_2\in R/R_x}||g(x)-g(x_2)||$, the second inequality is due to $x^-:=arg\min_{x_2\in R/R_x}||\hat{g}(x)-\hat{g}(x_2)||$. Analogously
  $$\min_{x_1\in R_x}||g(x)-g(x_1)||_2\leq ||g(x)-g(x^+)||_2\leq ||\hat{g}(x)-\hat{g}(x^+)||_2+2\sqrt{8F^2\ln(\frac{k+1}{\alpha})/{3n}}$$
  with probability at least $1-2\alpha$. Thus we have
  \begin{equation}
  \begin{aligned}
  d(x;g)&=\min_{x_2\in R/R_x} ||g(x)-g(x_2)||_2 -  \min_{x_1\in R_x}||g(x)-g(x_1)||_2\\
  &\geq||\hat{g}(x)-\hat{g}(x^-)||_2-2\sqrt{8F^2\ln(\frac{k+1}{\alpha})/{3n}}-\left(||\hat{g}(x)-\hat{g}(x^+)||_2+2\sqrt{8F^2\ln(\frac{k+1}{\alpha})/{3n}}\right)\\
  &=d(x,\hat{g}) - 4\sqrt{8F^2\ln(\frac{k+1}{\alpha})/{3n}}
  \end{aligned}
  \end{equation}
  with probability $1-4\alpha$. Replace $\alpha$ by $\frac{\alpha}{4}$ we obtain \autoref{prop:2}.
  \end{proof}
\section{Additional experiments on CUB200 and CARS196}\label{sec:additional experiments}
\noindent\textbf{Datasets.} We conduct experiments on two popular metric learning benchmarks: CUB200 and CARS196. We follow the setup in the previous work \cite{song2016deep} to split the training and test sets. 
\begin{itemize}
\vspace{-5pt}    
\item \textbf{CUB200-2011}
    contains 200 species of birds and 11,788 images \cite{WahCUB_200_2011}. We use the first 100 species as the training set and the rest as the test set.
\vspace{-5pt} 
\item \textbf{CARS196}
     has 196 models of cars and 16,185 images. \cite{KrauseStarkDengFei-Fei_3DRR2013}. We use the first 98 models as the training set and the rest as the test set.
\end{itemize}

\begin{figure*}[t]
		\begin{minipage}[t]{.5\linewidth}
			\centering
			\includegraphics[height=4cm]{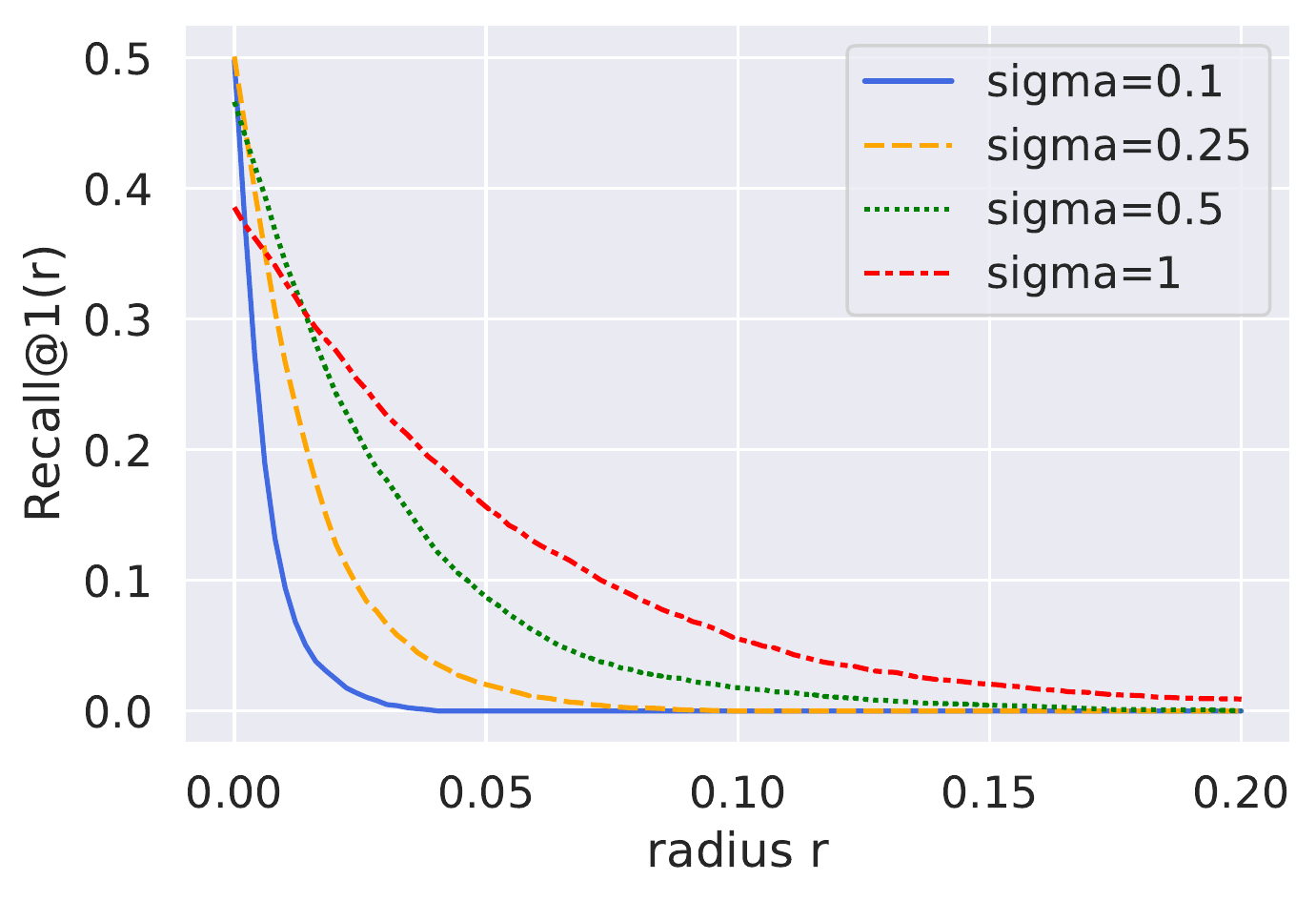}	
			\vspace{-6pt}
			\label{fig:margin_car}
		\end{minipage}
		\begin{minipage}[t]{.5\linewidth}
			\centering
			\includegraphics[height=4cm]{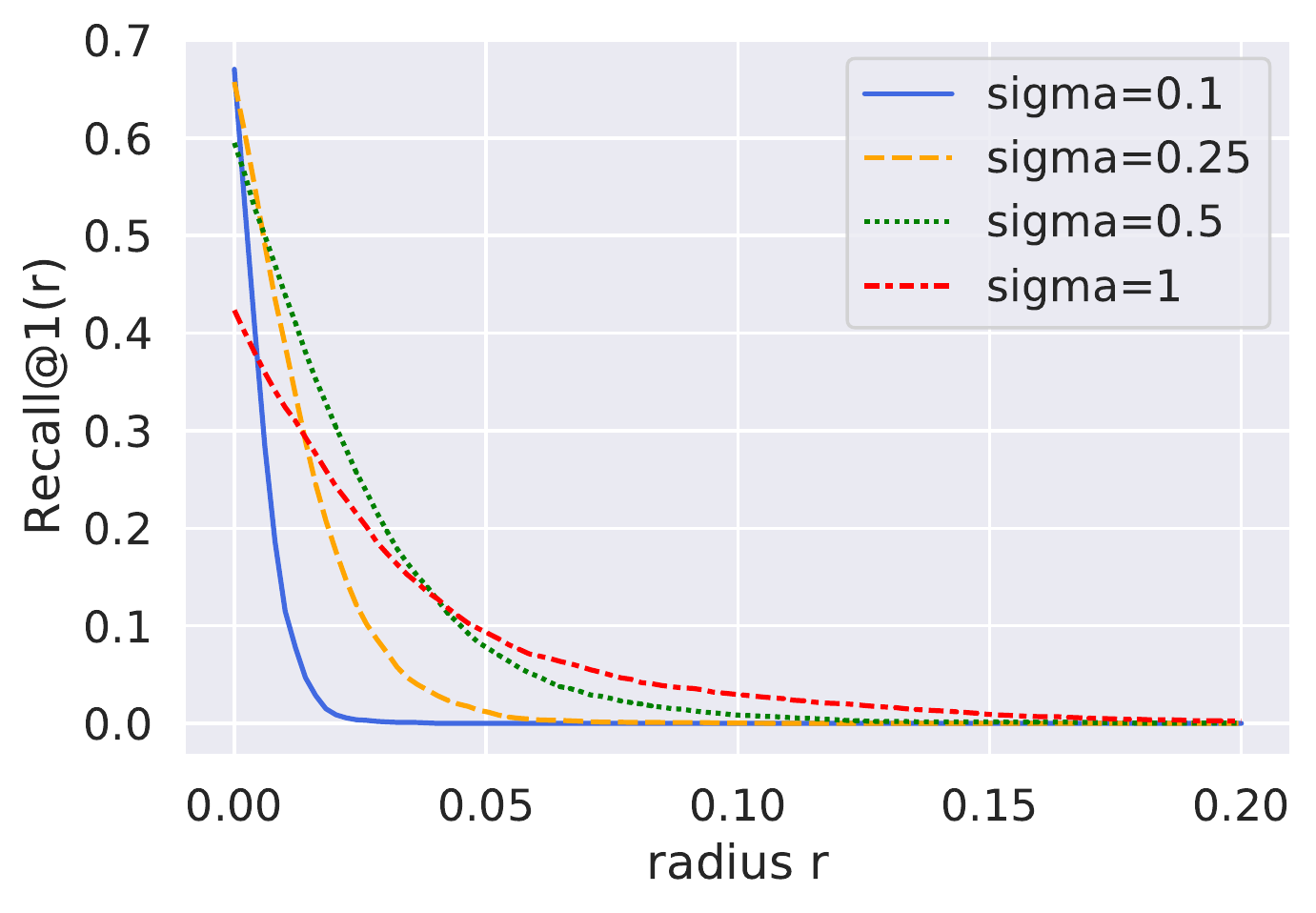}	
			\vspace{-6pt}
			\label{fig:ms_cub}
		\end{minipage}
	\vspace{-0.5cm}
        \caption{Experiments with GDML+RetrievalGuard on image retrieval benchmarks with different $\sigma$. \textbf{Left}: GDML+RetrievalGuard on CUB200-2011. \textbf{Right}: GDML+RetrievalGuard on CARS196.}
        \label{fig:GDML}
\end{figure*}

\end{document}